\documentclass[conference,a4paper]{IEEEtran}
\usepackage{amsmath}
\usepackage{amssymb}
\usepackage{amsfonts}
\usepackage{graphicx}
\usepackage{epsfig}
\usepackage{subfigure}
\usepackage{psfrag}
\usepackage{xcolor}

\begin{document}
\title{An Uplink-Downlink Duality for \\ Cloud Radio Access Network}

\author{ \IEEEauthorblockN{Liang Liu, Pratik Patil, and Wei Yu}
\IEEEauthorblockA{Department of Electrical and Computer Engineering\\
University of Toronto, Toronto, ON, M5S 3G4, Canada \\ Emails: lianguot.liu@utoronto.ca, $\{$ppatil,weiyu$\}$@comm.utoronto.ca }}

\maketitle

\begin{abstract}
Uplink-downlink duality refers to the fact that the Gaussian broadcast channel has the same capacity region as the dual Gaussian multiple-access channel under the same sum-power constraint. This paper investigates a similar duality relationship between the uplink and downlink of a cloud radio access network (C-RAN), where a central processor (CP) cooperatively serves multiple mobile users through multiple remote radio heads (RRHs) connected to the CP with finite-capacity fronthaul links. The uplink of such a C-RAN model corresponds to a multiple-access relay channel; the downlink corresponds to a broadcast relay channel. This paper considers compression-based relay strategies in both uplink and downlink C-RAN, where the quantization noise levels are functions of the fronthaul link capacities. If the fronthaul capacities are infinite, the conventional uplink-downlink duality applies. The main result of this paper is that even when the fronthaul capacities are finite, duality continues to hold for the case where independent compression is applied across each RRH in the sense that when the transmission and compression designs are jointly optimized, the achievable rate regions of the uplink and downlink remain identical under the same sum-power and individual fronthaul capacity constraints. As an application of the duality result, the power minimization problem in downlink C-RAN can be efficiently solved based on its uplink counterpart.
\end{abstract}

\IEEEpeerreviewmaketitle

\newtheorem{theorem}{Theorem}
\newtheorem{remark}{Remark}
\newcommand{\mv}[1]{\mbox{\boldmath{$ #1 $}}}

\section{Introduction}\label{sec:Introduction}
As a promising candidate for the future 5G standard, cloud radio access network (C-RAN) enables a centralized processing architecture, using multiple relay-like base stations (BSs), named remote radio heads (RRHs), to serve mobile users cooperatively under the coordination of a central processor (CP). The practically achievable throughput of C-RAN is largely constrained by the finite-capacity fronthaul links between the RRHs and the CP. In the literature, a considerable amount of effort has been dedicated to the study of efficient fronthaul techniques \cite{Yu16}. In the uplink, the compression-based strategy can be used, where each RRH samples, quantizes and forwards its received signals to the CP over its fronthaul link such that each user's messages can be jointly decoded \cite{Yuwei13}. Likewise, in the downlink, the CP can pre-form the beamforming vectors, then compress and transmit the beamformed signals via the fronthaul links to RRHs for coherent transmission \cite{Simeone13}. This paper focuses on these compression-based strategies and reveals an uplink-downlink duality for C-RAN.

Information theoretically, the uplink of C-RAN model corresponds to a multiple-access relay channel, while the downlink corresponds to a broadcast relay channel. When the fronthaul links have infinite capacities, the C-RAN model reduces to a multiple-access channel (MAC) in the uplink, and a broadcast channel (BC) in the downlink, for which there is a well-known uplink-downlink duality, i.e., the achievable rate regions with linear beamforming or the entire capacity regions of the MAC and the dual BC are identical under the same sum-power constraint \cite{Tassiulas98}--\cite{Yu07}. In this paper, we extend such an uplink-downlink duality relationship to the C-RAN model where the fronthaul capacities are finite. We show that if independent compression is applied across each RRH, the achievable rate regions of the uplink and downlink C-RAN are identical under the same sum-power and the same individual fronthaul capacity constraints when all the terminals are equipped with a single antenna. Specifically, given any downlink (uplink) power assignment, transmit (receive) beamforming and quantization noise levels, we construct the corresponding respective uplink (downlink) power assignment, receive (transmit) beamforming and quantization noise levels such that each user's access rate, each RRH's fronthaul rate, as well as the sum-power are all preserved. Based on this duality result, we further show that the downlink sum-power minimization problem can be efficiently solved based on its uplink counterpart.

Although this paper focuses on the compression-based strategies for C-RAN, we remark that other coding strategies can potentially outperform the compression strategy for certain channel parameters. For example, the  so-called ``compute-and-forward'' and ``reverse compute-and-forward'' strategies have been applied for the uplink and downlink C-RAN, respectively, where linear functions of transmitted codewords are computed and then sent between the RRHs and CP \cite{Caire13}. Moreover, in the downlink, the CP may opt to share user data directly with the RRHs \cite{Yu14}. However, the capacity of the uplink and downlink C-RAN model is still an open problem.

{\it Notation}: Scalars are denoted by lower-case letters, vectors
denoted by bold-face lower-case letters, and matrices denoted by
bold-face upper-case letters. $\mv{I}$ and $\mv{1}$  denote an
identity matrix and an all-one vector, respectively, with
appropriate dimensions. For a square matrix $\mv{S}$, $\mv{S}^{-1}$
denotes its inverse (if $\mv{S}$ is full-rank). For a matrix
$\mv{M}$ of arbitrary size, $[\mv{M}]_{i,j}$ denotes the entry on the $i$th row and $j$th column of $\mv{M}$. ${\rm
Diag}(x_1,\cdots,x_K)$ denotes a diagonal matrix
with the diagonal elements given by $x_1,\cdots,x_K$.  $\mathbb{C}^{x \times y}$ denotes the space of
$x\times y$ complex matrices.

\section{System Model}\label{sec:System Model}
\begin{figure}
\begin{center}
\scalebox{0.45}{\includegraphics*{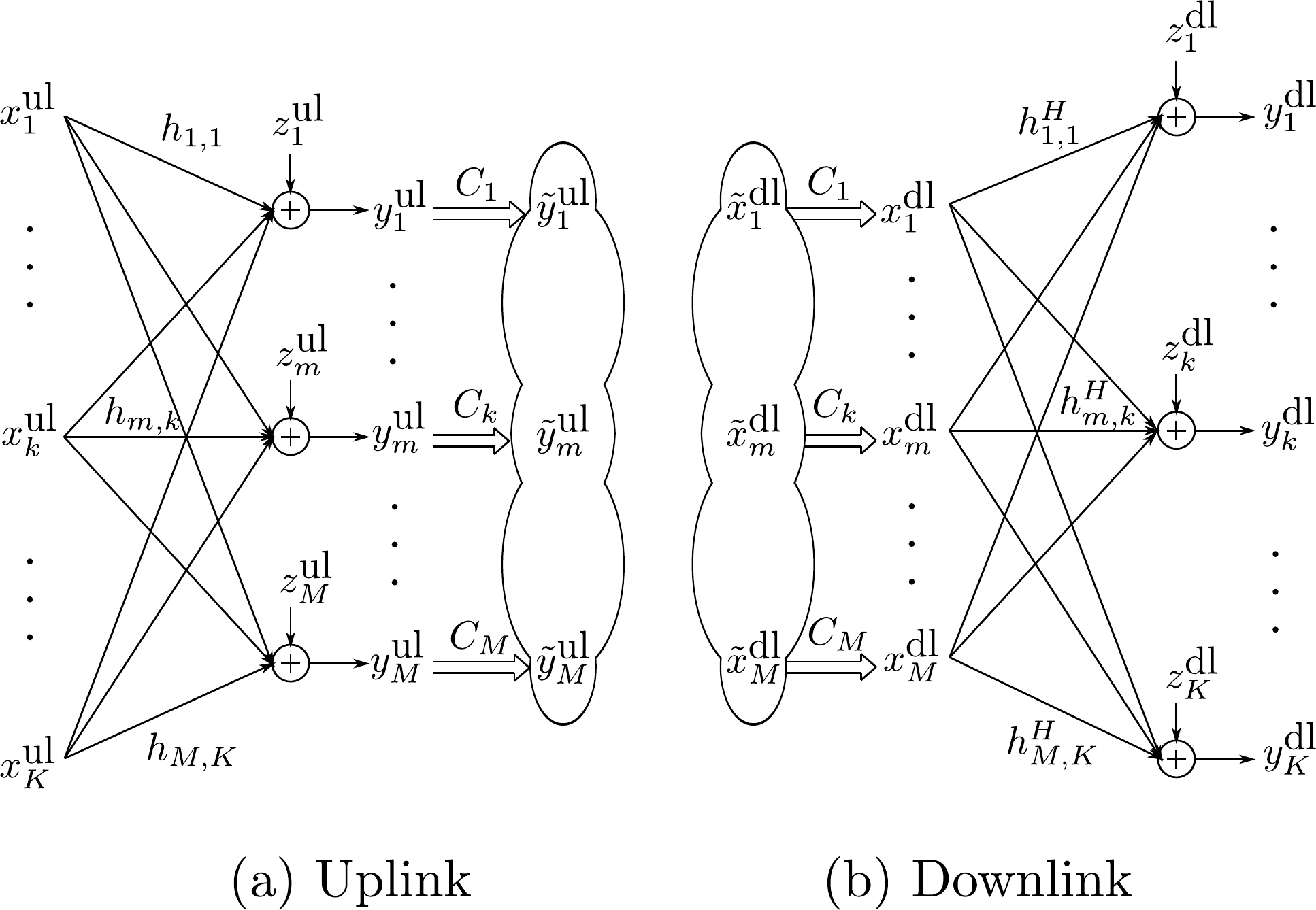}}
\end{center}
\caption{System model of the uplink and downlink C-RAN.}\label{fig2} \vspace{-15pt}
\end{figure}

The uplink and downlink models for a C-RAN consist of one CP, $M$ RRHs, and $K$ users. It is assumed that all the RRHs and users are equipped with a single antenna. In the uplink, the channel from user $k$ to RRH $m$ is denoted by $h_{m,k}$; in the dual downlink, the channel from  RRH $m$ to user $k$ is then given by $h_{m,k}^H$. The sum-power constraint is denoted by $P$ for both the uplink and downlink. It is assumed that RRH $m$ is connected to the CP via a noiseless digital fronthaul link with capacity $C_m$. In this paper, we focus on compression-based strategies to relay the information between the CP and RRHs via fronthaul links.


\subsection{Uplink C-RAN}\label{sec:Uplink C-RAN}
The uplink C-RAN model is as shown in Fig. \ref{fig2}(a). The discrete-time baseband channel between the users and RRHs can be modelled as
\begin{align}\label{eqn:uplink signal}
\hspace{-8pt} \left[\begin{array}{c} y_{1}^{\rm ul} \\ \vdots \\ y_M^{\rm ul} \end{array} \right]\hspace{-2pt} =\hspace{-2pt} \left[\begin{array}{ccc}h_{1,1} & \cdots & h_{1,K} \\  \vdots & \ddots & \vdots \\ h_{M,1} & \cdots & h_{M,K} \end{array}\right]\left[\begin{array}{c}x_1^{\rm ul} \\ \vdots \\ x_K^{\rm ul} \end{array}\right]\hspace{-2pt}+\hspace{-2pt}\left[\begin{array}{c} z_1^{\rm ul} \\ \vdots \\ z_M^{\rm ul} \end{array}\right],\end{align}
where $x_k^{\rm ul}$ denotes the transmit signal of user $k$, and $z_m^{\rm ul}\sim\mathcal{CN}(0,\sigma^2)$ denotes the additive white Gaussian noise (AWGN) at RRH $m$. In this paper, we assume that each of the users transmits using a Gaussian codebook, i.e., $x_k^{\rm ul}=\sqrt{p_k^{\rm ul}}s_k^{\rm ul}$, $\forall k$, where $s_k^{\rm ul}\sim \mathcal{CN}(0,1)$ denotes the message of user $k$, and $p_k^{\rm ul}$ denotes the transmit power of user $k$.

After receiving the wireless signals from users, RRH $m$ compresses $y_m^{\rm ul}$ and sends the compressed signals to the CP, $\forall m$. The quantization noise can be modeled as an independent Gaussian random variable, i.e,\begin{align}\label{eqn:uplink compression}
\tilde{y}_m^{\rm ul}=y_m^{\rm ul}+e_m^{\rm ul}=\sum\limits_{k=1}^Kh_{m,k}x_k^{\rm ul}+z_m^{\rm ul}+e_m^{\rm ul}, ~ \forall m,
\end{align}where $e_m^{\rm ul}\sim \mathcal{CN}(0,q_m^{\rm ul})$, and $q_m^{\rm ul}$ denotes the variance of the compression noise at RRH $m$. After receiving the compressed signals, the CP first decodes the compression codewords and then applies linear beamforming to decode each user's message, i.e.,
\begin{align}
\tilde{s}_k^{\rm ul}=\mv{w}_k^H\tilde{\mv{y}}^{\rm ul}, ~~~ \forall k,
\end{align}where $\mv{w}_k=[w_{k,1},\cdots,w_{k,M}]^T$ with $\|\mv{w}_k\|^2=1$ denotes the decoding beamforming vector for user $k$'s message, and $\tilde{\mv{y}}^{\rm ul}=[\tilde{y}_1^{\rm ul},\cdots,\tilde{y}_M^{\rm ul}]^T$ denotes the collective compressed signals from all RRHs.

The total transmit power of all the users is expressed as
\begin{align}\label{eqn:uplink sum-power}
P^{\rm ul}(\{p_i^{\rm ul}\})=\sum\limits_{i=1}^K\mathbb{E}[|x_i^{\rm ul}|^2]=\sum\limits_{i=1}^Kp_i^{\rm ul}.
\end{align}Moreover, in this paper we assume that the compression process is done independently across RRHs. Based on rate-distortion theory, the fronthaul rate for transmitting $\tilde{y}_m^{\rm ul}$ is expressed as
\begin{align}\label{eqn:uplink fronthaul rate}
\hspace{-5pt} C_m^{\rm ul}(\{p_i^{\rm ul}\},q_m^{\rm ul})=& ~ I(y_m^{\rm ul};\tilde{y}_m^{\rm ul}) \nonumber \\ = & ~ \log_2\frac{\sum\limits_{i=1}^Kp_i^{\rm ul}|h_{m,i}|^2+q_m^{\rm ul}+\sigma^2}{q_m^{\rm ul}}, ~ \forall m.
\end{align}Finally, with interference treated as noise, the achievable rate of user $k$ is expressed as
\begin{multline}\label{eqn:uplink rate}
R_k^{\rm ul}(\{p_i^{\rm ul},\mv{w}_i\},\{q_m^{\rm ul}\})=I(s_k^{\rm ul};\tilde{s}_k^{\rm ul}) \\ = \log_2\frac{\sum\limits_{i=1}^Kp_i^{\rm ul}|\mv{w}_k^H\mv{h}_i|^2+\sum\limits_{m=1}^Mq_m^{\rm ul}|w_{k,m}|^2+\sigma^2}{\sum\limits_{j\neq k}p_j^{\rm ul}|\mv{w}_k^H\mv{h}_j|^2+\sum\limits_{m=1}^Mq_m^{\rm ul}|w_{k,m}|^2+\sigma^2}, ~~~ \forall k,
\end{multline}where $\mv{h}_k=[h_{1,k},\cdots,h_{M,k}]^T$ denotes the collective channel from user $k$ to all RRHs.

Given the individual fronthaul capacity constraints $C_m$'s and sum-power constraint $P$,  define the set of feasible transmit power, compression noise levels, and receive beamforming vectors as
\begin{align}
& \mathcal{T}^{\rm ul}(\{C_m\},P)=\big\{(\{p_i^{\rm ul},\mv{w}_i\},\{q_m^{\rm ul}\}): P^{\rm ul}(\{p_i^{\rm ul}\})\leq P, \nonumber \\ & C_m^{\rm ul}(\{p_i^{\rm ul}\},q_m^{\rm ul})\leq C_m, \forall m,  \|\mv{w}_i\|^2=1, \forall i  \big\}.
\end{align}With independent compression and linear decoding, the achievable rate region in the uplink C-RAN is thus given by
\begin{align}\label{eqn:downlink rate region}
& \mathcal{R}^{\rm ul}(\{C_m\},P)  \triangleq
\bigcup\limits_{(\{p_i^{\rm ul},\mv{w}_i\},\{q_m^{\rm ul}\})\in \mathcal{T}^{\rm ul}(\{C_m\},P)} \nonumber \\ &  \left\{(r_1^{\rm ul},\cdots,r_K^{\rm ul}): r_k^{\rm ul} \leq R_k^{\rm ul}(\{p_i^{\rm ul},\mv{w}_i\},\{q_m^{\rm ul}\}),\forall k \right\}.
\end{align}

\subsection{Downlink C-RAN}\label{sec:Downlink C-RAN}
The downlink C-RAN model is as shown in Fig. \ref{fig2}(b). The discrete-time baseband channel model between the RRHs and the users is the dual of the uplink channel given by
\begin{align}\label{eqn:downlink received signal}
\hspace{-8pt} \left[\begin{array}{c} y_{1}^{\rm dl} \\ \vdots \\ y_K^{\rm dl} \end{array} \right]\hspace{-2pt} =\hspace{-2pt} \left[\begin{array}{ccc}h_{1,1}^H & \cdots & h_{M,1}^H \\  \vdots & \ddots & \vdots \\ h_{1,K}^H & \cdots & h_{M,K}^H \end{array}\right]\left[\begin{array}{c}x_1^{\rm dl} \\ \vdots \\ x_M^{\rm dl} \end{array}\right]\hspace{-2pt}+\hspace{-2pt}\left[\begin{array}{c} z_1^{\rm dl} \\ \vdots \\ z_K^{\rm dl} \end{array}\right],
\end{align}where $x_m^{\rm dl}$ denotes the transmit signal of RRH $m$, and $z_k^{\rm dl}\sim \mathcal{CN}(0,\sigma^2)$ denotes the AWGN at receiver $k$.

The transmit signals $x_m^{\rm dl}$'s are compressed versions of the beamformed signals $\tilde{x}_m^{\rm dl}$'s. Similar to (\ref{eqn:uplink compression}), the compression noise is modelled as a Gaussian random variable, i.e.,
\begin{align}\label{eqn:downlink compression}
x_m^{\rm dl}=\tilde{x}_m^{\rm dl}+e_m^{\rm dl}, ~~~ \forall m,
\end{align}where $e_m^{\rm dl}\sim \mathcal{CN}(0,q_m^{\rm dl})$, and $q_m^{\rm dl}$ denotes the variance of the compression noise at RRH $m$.

Similar to the uplink, a Gaussian codebook for each user is used for downlink transmission. Define the beamformed signal intended for user  $k$ to be transmitted across all the RRHs as $\mv{v}_k\sqrt{p_k^{\rm dl}}s_k^{\rm dl}$, $\forall k$, where $s_k^{\rm dl}\sim \mathcal{CN}(0,1)$ denotes the message for user $k$, $p_k^{\rm dl}$ denotes the transmit power, and $\mv{v}_k=[v_{k,1},\cdots,v_{k,M}]^T$ with $\|\mv{v}_k\|^2=1$ denotes the transmit beamforming vector. The aggregate signal intended for all the RRHs is thus given by $\sum_{i=1}^K\mv{v}_i\sqrt{p_i^{\rm dl}}s_i^{\rm dl}$, which is compressed and then sent to the RRHs via fronthaul links. The transmit signal across the RRHs is therefore expressed as
\begin{align}\label{eqn:downlink signal}
\left[\begin{array}{c} x_1^{\rm dl} \\ \vdots \\ x_M^{\rm dl}  \end{array}\right]=\left[\begin{array}{c} \sum_{i=1}^Kv_{i,1}\sqrt{p_i^{\rm dl}}s_i^{\rm dl} \\ \vdots \\ \sum_{i=1}^Kv_{i,M}\sqrt{p_i^{\rm dl}}s_i^{\rm dl}  \end{array}\right]+\left[\begin{array}{c}e_1^{\rm dl}\\ \vdots \\ e_M^{\rm dl}  \end{array}\right].
\end{align}Then, the transmit power of all the RRHs is expressed as
\begin{align}\label{eqn:downlink sum-power}
\hspace{-5pt} P^{\rm dl}(\{p_i^{\rm dl}\},\{q_m^{\rm dl}\})=\sum\limits_{m=1}^M\mathbb{E}[|x_m^{\rm dl}|^2]=\sum\limits_{i=1}^Kp_i^{\rm dl}+\sum\limits_{m=1}^Mq_m^{\rm dl}.
\end{align}Again, the compression is done independently across RRHs. As a result, the fronthaul rate for transmitting $x_m^{\rm dl}$ is expressed as
\begin{align}\label{eqn:downlink fronthaul rate}
\hspace{-5pt} C_m^{\rm dl}(\{p_i^{\rm dl},\mv{v}_i\},q_m^{\rm dl})= & I(\tilde{x}_m^{\rm dl};x_m^{\rm dl}) \nonumber \\ = & \log_2 \frac{\sum\limits_{i=1}^Kp_i^{\rm dl}|v_{i,m}|^2+q_m^{\rm dl} }{q_m^{\rm dl}}, ~~~ \forall m.
\end{align}Finally, with linear encoding, the achievable rate of user $k$ can be expressed as
\begin{multline}\label{eqn:downlink rate}
R_k^{\rm dl}(\{p_i^{\rm dl},\mv{v}_i\},\{q_m^{\rm dl}\})=I(s_k^{\rm dl};y_k^{\rm dl}) \\ = \log_2\frac{\sum\limits_{i=1}^Kp_i^{\rm dl}|\mv{v}_i^H\mv{h}_k|^2+\sum\limits_{m=1}^Mq_m^{\rm dl}|h_{m,k}|^2+\sigma^2}{\sum\limits_{j\neq k}p_j^{\rm dl}|\mv{v}_j^H\mv{h}_k|^2+\sum\limits_{m=1}^Mq_m^{\rm dl}|h_{m,k}|^2+\sigma^2}, ~~~ \forall k.
\end{multline}

Given the individual fronthaul capacity constraints $C_m$'s and sum-power constraint $P$,  define the set of feasible transmit power and beamforming vectors as well as compression noise levels as
\begin{align}
& \mathcal{T}^{\rm dl}(\{C_m\},P)= \big\{(\{p_i^{\rm dl},\mv{v}_i\},\{q_m^{\rm dl}\}): P^{\rm dl}(\{p_i^{\rm dl}\},\{q_m^{\rm dl}\})\leq P, \nonumber \\ & C_m^{\rm dl}(\{p_i^{\rm dl},\mv{v}_i\},q_m^{\rm dl}) \leq C_m, \forall m,    \|\mv{v}_i\|^2=1,\forall i \big\}.
\end{align}With independent compression and linear encoding, the achievable rate region in the downlink C-RAN is thus given by
\begin{align}\label{eqn:downlink rate region}
& \mathcal{R}^{\rm dl}(\{C_m\},P)\triangleq
\bigcup\limits_{(\{p_i^{\rm dl},\mv{v}_i\},\{q_m^{\rm dl}\})\in \mathcal{T}^{\rm dl}(\{C_m\},P)}  \nonumber \\ & \left\{ (r_1^{\rm dl},\cdots,r_K^{\rm dl}): r_k^{\rm dl} \leq R_k^{\rm dl}(\{p_i^{\rm dl},\mv{v}_i\},\{q_m^{\rm dl}\}), \forall k \right\}.
\end{align}

\section{Uplink-Downlink Duality for C-RAN}\label{sec:Uplink-downlink Duality}
In this section, we establish a duality relationship for C-RAN under the compression strategy by comparing the rate regions $\mathcal{R}^{\rm ul}(\{C_m\},P)$ and $\mathcal{R}^{\rm dl}(\{C_m\},P)$. As mentioned before, if the fronthaul links between the RRHs and the CP have infinite capacities, i.e., $C_m \rightarrow \infty$, $\forall m$, then the uplink and downlink C-RAN models reduce to the MAC and BC, respectively, thus the conventional uplink-downlink duality applies, i.e., $\mathcal{R}^{\rm ul}(\{C_m\rightarrow \infty\},P)=\mathcal{R}^{\rm dl}(\{C_m\rightarrow \infty\},P)$. With finite values of $C_m$'s, however, it is not obvious if duality still holds. On one hand, the compression noise contributes to the transmit power in the downlink, but has no contribution to the transmit power in the uplink. On the other hand, background Gaussian noise contributes to the fronthaul rate in the uplink, but it does not affect the fronthaul transmission in the downlink. Interestingly, the following theorem shows that these two effects counteract with each other and the uplink-downlink duality continues to exist in C-RAN even with finite fronthaul capacities.

\begin{theorem}\label{theorem1}
Consider a C-RAN model implementing compression-based strategies in both the uplink and the downlink, where all the users and RRHs are equipped with one single antenna. Then, any rate tuple that is achievable in the uplink is also achievable with the same sum-power and individual fronthaul capacity constraints in the downlink, and vice versa, i.e., $\mathcal{R}^{\rm ul}(\{C_m\},P)=\mathcal{R}^{\rm dl}(\{C_m\},P)$.
\end{theorem}

\begin{proof}
First, we show that given any feasible uplink solution $(\{\bar{p}_i^{\rm ul},\bar{\mv{w}}_i\},\{\bar{q}_m^{\rm ul}\})\in \mathcal{T}^{\rm ul}(\{C_m\},P)$, the following downlink problem is always feasible.
\begin{align}\hspace{-8pt} \mathop{\mathrm{find}} & ~ \{p_i^{\rm dl},\mv{v}_i\},\{q_m^{\rm dl}\}  \label{eqn:problem 2} \\
\hspace{-8pt} \mathrm {s.t.}  & ~R_k^{\rm dl}(\{p_i^{\rm dl},\mv{v}_i\},\{q_m^{\rm dl}\})=R_k^{\rm ul}(\{\bar{p}_i^{\rm ul},\mv{\bar{w}}_i\},\{\bar{q}_m^{\rm ul}\}), ~ \forall k, \label{eqn:downlink rate constraint} \\ \hspace{-8pt} & ~ C_m^{\rm dl}(\{p_i^{\rm dl},\mv{v}_i\},q_m^{\rm dl})=C_m^{\rm ul}(\{\bar{p}_i^{\rm ul}\},\bar{q}_m^{\rm ul}), ~ \forall m, \label{eqn:downlink fronthaul constraint} \\ \hspace{-8pt} & ~ P^{\rm dl}(\{p_i^{\rm dl}\},\{q_m^{\rm dl}\})=P^{\rm ul}(\{\bar{p}_i^{\rm ul}\}). \label{eqn:downlink power constraint}
\end{align}

According to (\ref{eqn:uplink fronthaul rate}), we have
\begin{align}\label{eqn:uplink quantization noise}
\bar{q}_m^{\rm ul}=\frac{\sum\limits_{i=1}^K\bar{p}_i^{\rm ul}|h_{m,i}|^2+\sigma^2}{\eta_m}, ~~~ \forall m,
\end{align}where $\eta_m=2^{C_m^{\rm ul}(\{\bar{p}_i^{\rm ul}\},\bar{q}_m^{\rm ul})}-1$. By substituting the quantization noise level from (\ref{eqn:uplink quantization noise}) in (\ref{eqn:uplink rate}), it follows that
\begin{align}\label{eqn:uplink SINR}
& \gamma_k = \nonumber \\ & \frac{\bar{p}_k^{\rm ul}|\bar{\mv{w}}_k^H\mv{h}_k|^2}{\sum\limits_{j\neq k}\bar{p}_j^{\rm ul}|\bar{\mv{w}}_k^H\mv{h}_j|^2+\sum\limits_{m=1}^M\frac{(\sum\limits_{i=1}^K\bar{p}_i^{\rm ul}|h_{m,i}|^2+\sigma^2)|\bar{w}_{k,m}|^2}{\eta_m}+\sigma^2}, ~~~ \forall k,
\end{align}where $\gamma_k=2^{R_k^{\rm ul}(\{\bar{p}_i^{\rm ul},\bar{\mv{w}}_i\},\{\bar{q}_m^{\rm ul}\})}-1$.

Define $\bar{\mv{p}}^{\rm ul}, \bar{\mv{\tau}}^{\rm ul}\in \mathbb{C}^{K\times 1}$ with the $k$th elements denoted as $\bar{p}_k^{\rm ul}$ and $\sum_{m=1}^M|\bar{w}_{k,m}|^2/\eta_m$, respectively, $\bar{\mv{D}}^{\rm ul}={\rm diag}(\gamma_1/|\bar{\mv{w}}_1^H\mv{h}_1|^2,\cdots,\gamma_K/|\bar{\mv{w}}_K^H\mv{h}_K|^2)$, and $\bar{\mv{A}}^{\rm ul}\in \mathbb{C}^{K\times K}$ with
\begin{align}
\left[ \bar{\mv{A}}^{\rm ul} \right]_{i,j}=\left\{\begin{array}{ll} \sum\limits_{m=1}^M\frac{|h_{m,i}|^2|\bar{w}_{i,m}|^2}{\eta_m}, & {\rm if} ~ i=j, \\ |\bar{\mv{w}}_i^H\mv{h}_j|^2+ \sum\limits_{m=1}^M\frac{|h_{m,j}|^2|\bar{w}_{i,m}|^2}{\eta_m}, & {\rm otherwise}. \end{array} \right. \label{eqn:uplink A}
\end{align}Then, we can rewrite (\ref{eqn:uplink SINR}) in the following matrix form:
\begin{align}\label{eqn:SINR}
\left(\mv{I}-\bar{\mv{D}}^{\rm ul}\bar{\mv{A}}^{\rm ul}\right)\bar{\mv{p}}^{\rm ul}=\sigma^2\bar{\mv{D}}^{\rm ul}\mv{1}+\sigma^2\bar{\mv{D}}^{\rm ul}\bar{\mv{\tau}}^{\rm ul}.
\end{align}Since $\bar{\mv{p}}^{\rm ul}>\mv{0}$ is a feasible solution to (\ref{eqn:SINR}) by choice, according to \cite[Theorem 2.1]{Seneta81}, it follows that $\rho(\bar{\mv{D}}^{\rm ul}\bar{\mv{A}}^{\rm ul})<1$, where $\rho(\cdot)$ denotes the spectral radius of the argument matrix.

Next, consider problem (\ref{eqn:problem 2}). It follows from constraint (\ref{eqn:downlink fronthaul constraint}) that
\begin{align}\label{eqn:downlink quantization noise}
q_m^{\rm dl}=\frac{\sum\limits_{i=1}^Kp_i^{\rm dl}|v_{i,m}|^2}{\eta_m}, ~~~ \forall m.
\end{align}By substituting the quantization noise level from (\ref{eqn:downlink quantization noise}) in (\ref{eqn:downlink rate constraint}), we have
\begin{align}\label{eqn:new SINR}
\hspace{-5pt}& \gamma_k= \nonumber \\ & \frac{p_k^{\rm dl}|\mv{v}_k^H\mv{h}_k|^2}{\sum\limits_{j\neq k}p_j^{\rm dl}|\mv{v}_j^H\mv{h}_k|^2+\sum\limits_{m=1}^M\frac{\sum\limits_{i=1}^Kp_i^{\rm dl}|v_{i,m}|^2|h_{m,k}|^2}{\eta_m}+\sigma^2}, ~ \forall k.
\end{align}We can rewrite (\ref{eqn:new SINR}) in the following matrix form:
\begin{align}\label{eqn:downlink SINR}
\left(\mv{I}-\mv{D}^{\rm dl}\mv{A}^{\rm dl}\right)\mv{p}^{\rm dl}=\sigma^2\mv{D}^{\rm dl}\mv{1},
\end{align}where $\mv{p}^{\rm dl}\in \mathbb{C}^{K\times 1}$ with the $k$th element denoted as $p_k^{\rm dl}$, $\mv{D}^{\rm dl}={\rm diag}(\gamma_1/|\mv{v}_1^H\mv{h}_1|^2,\cdots,\gamma_K/|\mv{v}_K^H\mv{h}_K|^2)$, and $\mv{A}^{\rm dl}\in \mathbb{C}^{K\times K}$ with
\begin{align}
\left[ \mv{A}^{\rm dl} \right]_{i,j}=\left\{\begin{array}{ll} \sum\limits_{m=1}^M\frac{|h_{m,i}|^2|v_{i,m}|^2}{\eta_m}, & {\rm if} ~ i=j, \\ |\mv{v}_j^H\mv{h}_i|^2+ \sum\limits_{m=1}^M\frac{|h_{m,i}|^2|v_{j,m}|^2}{\eta_m}, & {\rm otherwise}. \end{array} \right. \label{eqn:downlink A}
\end{align}

Next, we choose the downlink beamforming solution as
\begin{align}\label{eqn:downlink beamforming vector}
\mv{v}_k=\bar{\mv{w}}_k, ~~~ \forall k.
 \end{align}Then we have $\mv{D}^{\rm dl}=\bar{\mv{D}}^{\rm ul}$ and $\mv{A}^{\rm dl}=(\bar{\mv{A}}^{\rm ul})^T$. As a result, (\ref{eqn:downlink SINR}) reduces to
\begin{align}
\left(\mv{I}-\bar{\mv{D}}^{\rm ul}(\bar{\mv{A}}^{\rm ul})^T\right)\mv{p}^{\rm dl}=\sigma^2\bar{\mv{D}}^{\rm ul}\mv{1}.
\end{align}Since $\bar{\mv{D}}^{\rm ul}\bar{\mv{A}}^{\rm ul}$ and $\bar{\mv{D}}^{\rm ul}(\bar{\mv{A}}^{\rm ul})^T$ possess the same eigenvalues, we have $\rho(\bar{\mv{D}}^{\rm ul}(\bar{\mv{A}}^{\rm ul})^T)=\rho(\bar{\mv{D}}^{\rm ul}\bar{\mv{A}}^{\rm ul})<1$. According to \cite[Theorem 2.1]{Seneta81}, this implies the existence of a feasible power solution $\mv{p}^{\rm dl}>\mv{0}$ to equation (\ref{eqn:downlink SINR}), which is given by
\begin{align}\label{eqn:downlink power}
\mv{p}^{\rm dl}=\left(\mv{I}-\bar{\mv{D}}^{\rm ul}(\bar{\mv{A}}^{\rm ul})^T\right)^{-1}\sigma^2\bar{\mv{D}}^{\rm ul}\mv{1}.
\end{align}By substituting the above downlink power solution into (\ref{eqn:downlink quantization noise}), the corresponding quantization noise power levels can be obtained.

Given any feasible uplink solution $(\{\bar{p}_i^{\rm ul},\bar{\mv{w}}_i\},\{\bar{q}_m^{\rm ul}\})\in \mathcal{T}^{\rm ul}(\{C_m\},\bar{P})$, we have constructed a downlink solution as given by (\ref{eqn:downlink beamforming vector}), (\ref{eqn:downlink power}) and (\ref{eqn:downlink quantization noise}) which satisfies constraints (\ref{eqn:downlink rate constraint}) and (\ref{eqn:downlink fronthaul constraint}) in problem (\ref{eqn:problem 2}). Next, we show that this solution also satisfies the sum-power constraint (\ref{eqn:downlink power constraint}):
\begin{align}
\hspace{-5pt}& P^{\rm dl}(\{p_i^{\rm dl}\},\{q_m^{\rm dl}\})=\sum\limits_{i=1}^Kp_i^{\rm dl}+\sum\limits_{m=1}^Mq_m^{\rm dl} \nonumber \\ = & ~ \sum\limits_{i=1}^Kp_i^{\rm dl}+\sum\limits_{m=1}^M\sum\limits_{i=1}^K\frac{p_i^{\rm dl}|v_{i,m}|^2}{\eta_m}=\mv{1}^T\mv{p}^{\rm dl}+(\bar{\mv{\tau}}^{\rm ul})^T\mv{p}^{\rm dl} \nonumber \\ = & ~ \sigma^2(\mv{1}+\bar{\mv{\tau}}^{\rm ul})^T\left(\mv{I}-\bar{\mv{D}}^{\rm ul}(\bar{\mv{A}}^{\rm ul})^T\right)^{-1}\bar{\mv{D}}^{\rm ul}\mv{1} \nonumber \\ = & ~ \sigma^2\mv{1}^T\left(\mv{I}-\bar{\mv{D}}^{\rm ul}\bar{\mv{A}}^{\rm ul}\right)^{-1}\bar{\mv{D}}^{\rm ul}(\mv{1}+\bar{\mv{\tau}}^{\rm ul}) \nonumber \\  = & ~ \mv{1}^T\bar{\mv{p}}^{\rm ul}=\sum\limits_{i=1}^K\bar{p}_i^{\rm ul}=P^{\rm ul}(\{\bar{p}_i^{\rm ul}\}).
\end{align}

As result, given any feasible uplink solution $(\{\bar{p}_i^{\rm ul},\bar{\mv{w}}_i\},\{\bar{q}_m^{\rm ul}\})\in \mathcal{T}^{\rm ul}(\{C_m\},P)$, the constructed downlink solution given in (\ref{eqn:downlink beamforming vector}), (\ref{eqn:downlink power}) and (\ref{eqn:downlink quantization noise}) satisfies all the constraints in problem (\ref{eqn:problem 2}). In other words, any rate tuple that is achievable in the uplink is also achievable in the downlink.

Similarly, we can show that given any feasible downlink solution $(\{\bar{p}_i^{\rm dl},\bar{\mv{v}}_i\},\{\bar{q}_m^{\rm dl}\})\in \mathcal{T}^{\rm dl}(\{C_m\},P)$, the following uplink problem is always feasible.
\begin{align} \hspace{-5pt} \mathop{\mathrm{find}} & ~ \{p_i^{\rm ul},\mv{w}_i\},\{q_m^{\rm ul}\}  \label{eqn:problem 1} \\
\hspace{-5pt}\mathrm {s.t.}  & ~R_k^{\rm ul}(\{p_i^{\rm ul},\mv{w}_i\},\{q_m^{\rm ul}\})=R_k^{\rm dl}(\{\bar{p}_i^{\rm dl},\bar{\mv{v}}_i\},\{\bar{q}_m^{\rm dl}\}),  \forall k, \label{eqn:uplink rate constraint} \\ \hspace{-5pt} & ~ C_m^{\rm ul}(\{p_i^{\rm ul}\},q_m^{\rm ul})=C_m^{\rm dl}(\{\bar{p}_i^{\rm dl},\bar{\mv{v}}_i\},\bar{q}_m^{\rm dl}), ~ \forall m, \label{eqn:uplink fronthaul constraint} \\ \hspace{-5pt} & ~ P^{\rm ul}(\{p_i^{\rm ul}\})=P^{\rm dl}(\{\bar{p}_i^{\rm dl}\},\{\bar{q}_m^{\rm dl}\}). \label{eqn:uplink power constraint}
\end{align}In other words, any rate tuple that is achievable in the downlink is also achievable in the uplink. Theorem \ref{theorem1} is thus proved.

\end{proof}

\begin{remark}\label{remark1}
Besides linear encoding and decoding, to achieve larger rate regions, dirty-paper coding can be used in the downlink and successive interference cancellation can be used in the uplink at the CP. Consider again the compression-based strategy C-RAN with single-antenna terminals. Using similar proof technique as for Theorem \ref{theorem1}, it can be shown that by reversing the encoding and decoding order, any rate tuple that is achievable in the uplink is achievable with the same sum-power and individual fronthaul capacity constraints in the downlink, and vice versa. As a result, uplink-downlink duality also applies with non-linear encoding and decoding techniques.
\end{remark}

\begin{remark}\label{remark2}
In \cite{Yu07}, it is shown based on the minimax optimization technique that the capacity region of the BC with per-antenna power constraints is the same as that of its dual MAC with an uncertain noise at the receiver constrained by a certain positive semidefinite covariance matrix. Following a similar approach, it can be shown that the achievable rate region of the downlink C-RAN with per-RRH power constraints is the same as that of its dual uplink C-RAN with uncertain noises across all the RRHs.
\end{remark}

\section{Application of Uplink-Downlink Duality}\label{sec:Application of Uplink-Downlink Duality}
This section illustrates an application of uplink-downlink duality. We first show that the sum-power minimization problem in the uplink C-RAN can be globally solved by the celebrated fixed-point method \cite{Yates95}. Then, the optimal solution to the downlink sum-power minimization problem is obtained based on the uplink solution. Specifically, in the uplink, the sum-power minimization problem is formulated as
\begin{align}\mathop{\mathrm{minimize}}_{\{p_i^{\rm ul},\mv{w}_i\},\{q_m^{\rm ul}\}} & ~~~ P^{\rm ul}(\{p_i^{\rm ul}\}) \label{eqn:problem P1}  \\
\mathrm {subject \ to} ~ & ~~~R_k^{\rm ul}(\{p_i^{\rm ul},\mv{w}_i\},\{q_m^{\rm ul}\})\geq R_k, ~~~ \forall k,  \nonumber \\ & ~~~C_m^{\rm ul}(\{p_i^{\rm ul}\},q_m^{\rm ul}) \leq C_m, ~~~ \forall m, \nonumber \end{align}where $R_k$ denotes the rate requirement of user $k$. It can be shown that with the optimal solution, the fronthaul capacities should be fully used, i.e., (\ref{eqn:uplink quantization noise}) holds with $\eta_m=2^{C_m}-1$, $\forall m$. Moreover, the optimal beamforming solution is the well-known minimum-mean-square-error (MMSE) based receiver, i.e.,
\begin{align}\label{eqn:mmse receiver}
\mv{w}_k=\left(\sum\limits_{j\neq k}p_j^{\rm ul}\mv{h}_j\mv{h}_j^H+\mv{Q}^{\rm ul}(\mv{p}^{\rm ul})+\sigma^2\mv{I}\right)^{-1}\mv{h}_k, ~ \forall k,
\end{align}where $\mv{p}^{\rm ul}\in \mathbb{C}^{K\times 1}$ with the $k$th element denoted by $p_k^{\rm ul}$, and $\mv{Q}^{\rm ul}(\mv{p}^{\rm ul})={\rm diag}(q_1^{\rm ul},\cdots,q_M^{\rm ul})$ is a diagonal matrix with $q_m^{\rm ul}$'s given by (\ref{eqn:uplink quantization noise}). By substituting the quantization noise levels and receive beamforming vectors with (\ref{eqn:uplink quantization noise}) and (\ref{eqn:mmse receiver}), problem (\ref{eqn:problem P1}) reduces to the following power control problem
\begin{align}\mathop{\mathrm{minimize}}_{\mv{p}^{\rm ul}} ~ & ~~~ P^{\rm ul}(\{p_i^{\rm ul}\})  \label{eqn:problem P1'} \\
\mathrm {subject \ to} & ~~~\mv{p}^{\rm ul}\geq \mv{\Gamma}(\mv{p}^{\rm ul}), \nonumber \end{align}where $\mv{\Gamma}(\mv{p}^{\rm ul})\in \mathbb{C}^{K\times 1}$ with the $k$th element as
\begin{align}\label{eqn:gamma}
& [\mv{\Gamma}(\mv{p}^{\rm ul})]_k\nonumber \\ = & \frac{2^{R_k}-1}{\mv{h}_k^H\left(\sum\limits_{j\neq k}p_j^{\rm ul}\mv{h}_j\mv{h}_j^H+\mv{Q}^{\rm ul}(\mv{p}^{\rm ul})+\sigma^2\mv{I}\right)^{-1}\mv{h}_k}, ~ \forall k.
\end{align}It can be shown that $\mv{\Gamma}(\mv{p}^{\rm ul})$ is a standard interference function \cite{Yates95}. According to \cite[Theorem 2]{Yates95}, the following fixed-point algorithm can converge to the globally optimal power solution to problem (\ref{eqn:problem P1'}) given any initial point:
\begin{align}
\mv{p}^{{\rm ul},(n+1)}=\mv{\Gamma}(\mv{p}^{{\rm ul},(n)}),
\end{align}where $\mv{p}^{{\rm ul},(n)}$ denotes the power allocation solution obtained in the $n$th iteration of the above method. After the optimal power allocation is obtained, the optimal quantization noise levels and receive beamforming vectors to problem (\ref{eqn:problem P1}) can be obtained according to (\ref{eqn:uplink quantization noise}) and (\ref{eqn:mmse receiver}), respectively.

Next, consider the sum-power minimization problem in the downlink.
\begin{align}\mathop{\mathrm{minimize}}_{\{p_i^{\rm dl},\mv{v}_i\},\{q_m^{\rm dl}\}} & ~~~ P^{\rm dl}(\{p_i^{\rm dl}\},\{q_m^{\rm dl}\}) \label{eqn:problem P2}  \\
\mathrm {subject \ to} ~ & ~~~R_k^{\rm dl}(\{p_i^{\rm dl},\mv{v}_i\},\{q_m^{\rm dl}\})\geq R_k, ~~~ \forall k,  \nonumber \\ & ~~~C_m^{\rm dl}(\{p_i^{\rm dl},\mv{v}_i\},q_m^{\rm dl}) \leq C_m, ~~~ \forall m. \nonumber \end{align}In general, the downlink sum-power minimization problem is more involved since unlike the uplink, the optimal transmit beamforming vectors are not easy to obtain. However, Theorem \ref{theorem1} indicates that the optimal values of problems (\ref{eqn:problem P1}) and (\ref{eqn:problem P2}) are identical. As a result, we can find the optimal solution to problem (\ref{eqn:problem P1}) first, based on which we can then construct the optimal downlink solution according to (\ref{eqn:downlink beamforming vector}), (\ref{eqn:downlink power}), and (\ref{eqn:downlink quantization noise}).

\begin{remark}\label{remark3}
Similar to the power minimization problem, uplink-downlink duality can be applied to solve the downlink weighted sum-rate maximization problem with per-RRH power constraints based on the dual uplink, which is typically easier to deal with (see Remark \ref{remark2}).
\end{remark}

\section{Conclusion}\label{sec:Conclusion}
An uplink-downlink duality relationship is established for C-RAN. Specifically, we show that if transmission and compression designs are jointly optimized, the achievable rate regions of the uplink and downlink C-RAN are identical under the same sum-power and individual fronthaul capacity constraints when independent compression is performed across RRHs. Furthermore, this duality result proves useful for solving the downlink power minimization problem based on its dual problem in the uplink.

\end{document}